\documentclass[usenames,dvipsnames]{llncs}

\usepackage{amsmath,amssymb}
\usepackage{booktabs}
\usepackage{graphicx}
\usepackage{subfigure}
\usepackage{float}
\graphicspath{{img/}}

\newcommand{\pre}[1]{{^\bullet{#1}}}
\newcommand{\post}[1]{{#1}^\bullet}
\newcommand{\neighb}[1]{^\bullet{#1}^\bullet}
\newcommand{\inp}[1]{{^\bigcirc{#1}}}
\newcommand{\outp}[1]{#1^\bigcirc}
\newcommand{\reach}[1]{\lbrack#1\rangle}
\newcommand{\om}[1]{\varphi(#1)}
\newcommand{\abs}[1]{\lvert #1 \rvert}
\newcommand{\act}{\mathcal{A}}
\newcommand{\nat}{\mathbb{N}}

\title{Compositional Discovery of Workflow Nets from Event Logs Using Morphisms\thanks{This work is supported by MIUR, Basic Research Program at the National Research University Higher School of Economics, and Russian Foundation for Basic Research, project No.16-01-00546.}}
\author{Luca Bernardinello\inst{2}, Irina Lomazova \inst{1} \and  Roman Nesterov\inst{1,2} \and Lucia Pomello\inst{2}}
\institute{National Research University Higher School of Economics, \\20 Myasnitskaya Ulitsa, 101000 Moscow, Russia  \and Dipartimento di Informatica Sistemistica e Comunicazione, \\
Universit\`{a} degli Studi di Milano-Bicocca, \\
	Viale Sarca 336 - Edificio U14, I-20126 Milano, Italia}

\begin{document}
	\maketitle
	
	\begin{abstract}
		This paper presents a modular approach to discover process models for multi-agent systems from event logs. System event logs are filtered according to individual agent behavior. We discover workflow nets for each agent using existing process discovery algorithms. We consider asynchronous interactions among agents. Given a specification of an interaction protocol, we propose a general scheme of workflow net composition. By using morphisms, we prove that this composition preserves soundness of components. A quality evaluation shows the increase in the precision of models discovered by the proposed approach.
		
	\end{abstract}
	
	\section{Introduction}
Process discovery focuses on the synthesis of process models from event logs containing the observed record of an information system behavior. 
Process models are usually developed at the design stage of the information system life-cycle. 
However, the real observed behavior of the information system can eventually differ from the designed one. 
In some cases, designers cannot develop precise models describing all possible scenarios. 
That is why, process discovery is a topic of great interest at the moment.
	
Many process discovery algorithms have been proposed over recent years. 
They include Genetic algorithms, HeuristicMiner, Fuzzy miner, Inductive miner, the algorithms based on integer linear programming (ILP), and on the theory of regions (see \cite{Augusto2017} for a comprehensive review). 
They can be applied to solve typical problems of event logs, e.g. incompleteness and noise.
	
Within multi-agent systems (MAS), models obtained by existing process discovery tools can be incomprehensible since concurrent interacting agents produce rather sophisticated behavior. 
Process models of MAS, obtained by the algorithms mentioned above, are not structured in such a way that it is possible to identify agent models as components as shown by the following example.

In order to give an intuition behind our approach, consider the system model shown in Fig. \ref{exintro}(a). Two interacting agents can be clearly identified as well as the way they interact by looking at the small grey places and arcs. Now take the event log produced by this system as the input to a process discovery algorithm. If we try to discover a model from this log directly, we can obtain, for instance, the results shown in Fig. \ref{exintro}(b) by \emph{inductive miner} and in Fig. \ref{exintro}(c) by \emph{ILP miner}. 
Although equivalent to the original one in their ability to reproduce the same event sequences, their structure hides the fact the original system is made of two agents, communicating through channels. 
More technically, the two agents correspond to two S-invariants in the original model, while they are not ``separable'' by means of S-invariants in the discovered models.
It is true that we can improve the overall structure of models by configuring algorithm parameters. 
However, they still will not reflect the underlying MAS organization.
\begin{figure}
	\vspace{-0.5cm}
	\centering
	\subfigure[an intitial system model]{\includegraphics[height=4cm]{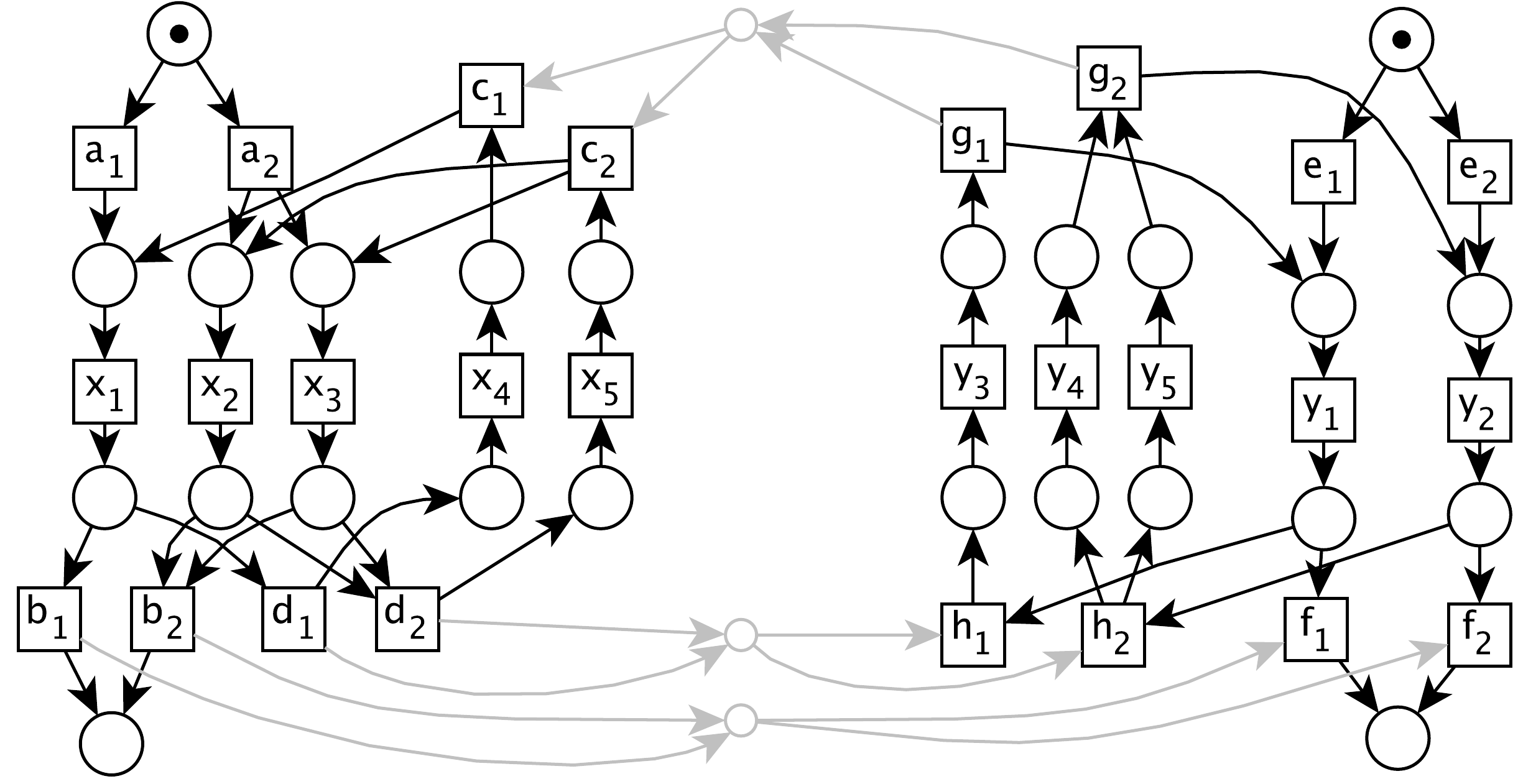}}
	\subfigure[inductive miner]{\includegraphics[height=4.8cm]{masintro_induct}}\hspace{0.5cm}
	\subfigure[ILP miner]{\includegraphics[height=5.4cm]{masintro_ilp}}
	\caption{Discovering process models of MAS from event logs}
	\label{exintro}
	\vspace{-0.65cm}
\end{figure}


We propose a compositional approach to address the problem of discovering process models of MAS from event logs clearly expressing agents as components and their interactions. 
We assume that agents interact asynchronously via message passing. 
System event logs are projected on each agent to discover component models in terms of workflow nets by using existing process discovery algorithms.
By means of morphisms, we can construct abstract models of the agents and compose them by adding channels for message passing. 
This composition models the protocol of agent interaction. 
In this way, it is possible to check soundness of this simplified model and to compose the discovered agent models on the basis of this protocol. 
We prove that this composition preserves soundness of the protocol and components by construction.
In this paper, we start constructing a \emph{formal background} to the general compositional approach to discovering process models of multi-agent systems from event logs.

		
The problem of discovering structured models from logs is not new and has been studied in several works based on composition. In \cite{Kalenkova14} the authors have designed a technique to discover readable models by decomposing transition systems. A special method to deal with process cancellations has been studied in \cite{Kalenkova14-1}. Regular process behavior is composed with cancellations using reset arcs.

A rather large amount of literature has been devoted to Petri net composition (see, for example, \cite{BoxCalc,Valk2003,Reisig2013}).
In particular, regarding asynchronous communication, several approaches have been proposed. 
Among the others, in \cite{Haddad13} asynchronous composition of Petri nets through channels has been studied considering preservation of various channel properties.
A general approach to asynchronous composition has also been discussed in \cite{Baldan01}, where open nets have been introduced.
The core problem of composition lies in preserving component properties.
Open nets have been used by many researches especially focused on modeling composite services (see, for example, \cite{vanHee2011,vanHee2010}).
They have proposed both structural and behavioral techniques assuring the correctness of composition.

Using morphisms is another possible way to achieve inheritance of component properties.
Composition of Petri nets via morphisms has been studied in several works \cite{Baldan01,Bednarczyk03,Nhat,Fabre06,Nielsen92,Padberg2003,Winskel87}.
We \kern -0.07em will \kern -0.05em use $\alpha$-morphisms and the related composition presented in \cite{Bernardinello2013}.
In general, $\alpha$-morphisms allow one to preserve/reflect properties checking structural and only \emph{local} behavioral constraints. 
In the particular case considered here, $\alpha$-morphisms preserve/reflect reachable markings and transition firings as well as preserve soundness as it will be shown in Section 3.


From a more practical point of view, many researchers have considered workflow net composition. 
Workflow nets (WF-nets) form a class of Petri nets used to model processes and services. 
Among the others, composition of WF-nets via shared resources has been studied in \cite{Lomazova13,WFRes} with a concern about soundness preservation. 
In addition, other approaches have been used for modeling and composing interacting workflow nets \cite{Lomazova10,Pankratius05}.
Many works have concerned the web service composition (see the survey \cite{Cardinale13}), where the authors have stressed that there is a lack of execution engines based on Petri nets.

The paper is organized as follows. 
Section 2 gives preliminary definitions used in this paper as well as composition by channels and $\alpha$-morphisms. 
Section 3  discusses properties preserved and reflected by $\alpha$-morphisms, relevant to WF-nets. 
Section 4 introduces a modular approach to construct models of MAS from event logs, using the same \emph{illustrative} example shown in Fig. \ref{exintro}.
In Section 5, we summarize the paper by discussing results and possible continuations.


%

	
\section{Preliminaries}
\vspace{-0.25cm}
$\nat$ denotes the set of non-negative integers. $A^+$ denotes the set of all finite non-empty sequences over $A$, and $A^* = A^+\cup \{\epsilon\}$, $\epsilon$ is the empty sequence.
A \textit{multiset} $m$ over a set $S$ is a function $m \colon S \rightarrow \nat$.
Let $m_1, m_2$ be two multisets over $S$. 
Then $m_1 \subseteq m_2 \Leftrightarrow m_1(s)\leq m_2(s)$ for all $s \in S$. 
Also, $m'=m_1+m_2 \Leftrightarrow m'(s)=m_1(s)+m_2(s)$, $m''=m_1-m_2 \Leftrightarrow m''(s)=\max(m_1(s)-m_2(s), 0)$ for all $s\in S$.  

A \textit{Petri net} is a triple $N=(P, T, F)$, where $P$ and $T$ are two disjoint sets of places and transitions, i.e. 
$P \cap T = \varnothing$, and $F \subseteq (P \times T) \cup (T \times P)$   
is a flow relation, where $ dom(F)  \cup cod(F) = P \cup T$ . We consider nets, s.t. $\forall t \in T \colon \abs{\pre{t}} \geq 1$ and $\abs{\post{t}} \geq 1$.

Let $N=(P, T, F)$ be a Petri net, and $X = P \cup T$. The set $\pre{x} = \{y \in X \vert (y, x) \in F\}$ denotes the \textit{preset} of $x \in X$. The set $\post{x} = \{y \in X \vert (x, y) \in F\}$ denotes the \textit{postset} of $x \in X$. The set $\neighb{x} = \pre{x} \cup \post{x}$ denotes the \textit{neighborhood} of $x \in X$.
Let $A \subseteq X$, then $\pre{A} = \bigcup_{x \in A}\pre{x}$, $\post{A} = \bigcup_{x \in A}\post{x}$, $\neighb{A} = \pre{A} \cup \post{A}$.

 By $N(A)$ we denote a \textit{subnet} of $N$ \textit{generated by} $A$, i.e. $N(A) = (P\cap A, T \cap A, F \cap (A \times A))$. 
 Let $N(A)$ be a subnet of $N$ generated by $A \subseteq X$. 
 The set $\inp{N(A)} = \{y \in A \vert \exists z \in X \setminus A \colon (z, y) \in F \text{ or } \pre{y} = \varnothing\}$ denotes the \textit{input} elements. 
 The set $\outp{N(A)} = \{y \in A \vert \exists z \in X \setminus A \colon (y, z) \in F \text{ or } \post{y} = \varnothing\}$ denotes the \emph{output} elements.


A \textit{marking} (state) of a Petri  net $N=(P, T, F)$ is a function $m \colon P \to \nat$. It is a multiset over a set of places $P$. A \textit{marked} Petri net $(N, m_0)$ is a Petri net together with its initial marking $m_0$.
 A marking $m$ \textit{enables} a transition $t \in T$, denoted $m\reach{t}$, if $\pre{t} \subseteq m$. The \textit{firing} of $t$ at $m$ leads to a new marking $m' = m - \pre{t} + \post{t}$, denoted $m\reach{t}m'$.

A sequence $w \in T^*$ is a \emph{firing sequence} of $N =(P, T, F, m_0)$ if $w=t_1t_2\dots t_n$ and $m_0\reach{t_1} m_1 \reach{t_2}\dots m_{n-1}\reach{t_n} m_n$. Then we can write $m_0\reach{w}m_n$. The set of all firing sequences of $N$ is denoted by $FS(N)$. 

A marking $m$ of $N =(P, T, F, m_0)$ is \emph{reachable} if $\exists w \in FS(N) \colon m_0\reach{w}m$. 
The set of all markings of $N$ reachable from $m$ is denoted by $\reach{m}$. A reachable marking is \emph{dead} if it does not enable any transition. $N$ is \emph{deadlock-free} if no reachable marking is dead. 
$N$ is \emph{safe} if $\forall p \in P\, \forall m \in \reach{m_0} \colon m(p) \leq 1$. Then we will specify reachable markings as subsets of places.


A \emph{state machine} is a connected Petri net $N =(P, T, F)$,  s.t. $\forall t \in T\colon \abs{\pre{t}}=\abs{\post{t}}=1$. 
A subnet of a marked Petri net $N = (P, T, F, m_0)$ identified by a subset of places $A \subseteq P$ and its neighborhood, i.e. 
$N(A \cup {\neighb{A}})$, is a \emph{sequential component} of $N$ if it is a state machine and has a single token in the initial marking. 
$N$ is \emph{covered} by sequential components if every place of $N$ belongs to at least one sequential component. 
Then $N$ is said to be \emph{state machine decomposable} (SMD).

Semantics of a marked Petri net is given by its \emph{unfolding} as defined below.

	Let $N = (P, T, F)$ be a Petri net, and $F^*$ be the reflexive transitive closure of F. Then $ \forall x, y \in P \cup T \colon$
		(a) $x$ and y are in \emph{causal} relation, denoted $x < y$, if $(x, y) \in F^*$; (b) $x$ and $y$ are in \emph{conflict} relation, denoted $x \# y$, if $\exists t_x, t_y \in T$, s.t. $t_x \neq t_y$, $\pre{t_x} \cap \pre{t_y} \neq \varnothing$, and $t_x < x$, $t_y < y$.

\begin{definition}
	A Petri net $O = (B, E, F)$ is an occurrence net if:
	\begin{enumerate}
		\item $\forall b \in B \colon \abs{\pre{b}} \leq 1$.
		\item $F^*$ is a partial order.
		\item $\forall x \in B \cup E \colon \{y \in B \cup E \vert y <x\}$ is finite.
		\item $\forall x, y \in B \cup E \colon x \# y \Rightarrow x \neq y$.
	\end{enumerate}
\end{definition}

By definition, $O$ is acyclic. $Min(O)$ denotes the set of minimal nodes of $O$ w.r.t. $F^*$, i.e. the elements having the empty preset. 
We only consider nets having transitions with non-empty presets and postsets, then $Min(O) \subseteq B$.

\begin{definition}
	Let $N=(P, T, F, m_0)$ be a marked safe Petri net, $O=(B, E, F)$ be an occurrence net, and $\pi:B\cup E\to P\cup T$ be a map. $(O, \pi)$ is a branching process of $N$ if:
	\vspace{-0.2cm}
	\begin{enumerate}
		\item $\pi(B) \subseteq P$ and $\pi(E) \subseteq T$.
		\item $\pi$ restricted to Min(O) is a bijection between Min(O) and $m_0$.
		\item $\forall t \in T\colon \pi$ restricted to $\pre{t}$ is a bijection between $\pre{t}$ and $\pre{\pi(t)}$, and similarly for $\post{t}$ and $\post{\pi(t)}$.
		\item $\forall t_1, t_2 \in T \colon$ if $\pre{t_1} = \pre{t_2}$ and $\pi(t_1) = \pi(t_2)$, then $t_1 = t_2$.
	\end{enumerate}
\end{definition}

The \emph{unfolding} of $N$, denoted $\mathcal{U}(N)$, is the maximal branching process of $N$, 
s.t. any other branching process of $N$ is isomorphic to a 
subnet of $\mathcal{U}(N)$ with the map $\pi$ restricted to the elements of this subnet. The map associated with the unfolding is denoted $u$ and called \emph{folding}.

\emph{Workflow nets} form a subclass of Petri nets used for process modeling. 
We define generalized workflow nets having an initial state  $m_0$ and a final state $m_f$. 

\begin{definition}\label{WF}
	A marked Petri net $N = (P, T, F, m_0, m_f)$ is a generalized workflow net (GWF-net) if and only if:
	\begin{enumerate}
		\item $m_0 = \{s \in P \,\vert\, \pre{s} = \varnothing\}$ and $m_0 \neq \varnothing$.
		\item $m_f = \{f \in P \,\vert\, \post{f} = \varnothing \}$ and $m_f \neq \varnothing$.
		\item $\forall x \in P \cup T \,\,\exists s \in m_0 \,\,\exists f \in m_f \colon (s, x) \in F^* \text{ and } (x, f) \in F^*$.
	\end{enumerate}
\end{definition}

If $\abs{m_0} = \abs{m_f} = 1$, then a generalized workflow net is called just  a workflow net (WF-net, for short). 
State machine decomposable GWF-nets are safe.
The important correctness property of GWF-nets is \emph{soundness} \cite{Aalst11}. 

\begin{definition}\label{sound}
	A GWF-net $N = (P, T, F, m_0, m_f)$ is sound if and only if:
	\begin{enumerate}
		\item $\forall m \in \reach{m_0} \colon m_f \in \reach{m}$.
		\item $\forall m \in \reach{m_0} \colon m_f \subseteq m \Rightarrow m = m_f$. 
		\item $\forall t \in T \, \exists m \in \reach{m_0} \colon m\reach{t}$.
	\end{enumerate}
\end{definition}

We consider composition of two  GWF-nets by adding \emph{channels}. 
A set of channel places (channels, for short) is denoted by $P_c$.
This approach is similar to the one presented in \cite{Haddad13}.
They model asynchronous communication via message passing.  
Some transitions of two nets can either \emph{send} a message by an incoming arc to a channel or \emph{receive} a message by an outgoing arc from a channel.
We assume to know exactly which transitions send/receive messages to/from which channel places.
In order to simplify the notation, we will not introduce special transition labels indicating sending/receiving.

Two GWF-nets can be composed via a set of channels $P_c$ iff any channel receiving a message from one GWF-net send it only to the other GWF-net.


\begin{definition}\label{netplus}
	Let $N_i = (P_i, T_i, F_i, m_0^i, m_f^i)$ be a GWF-net for $i=1, 2$, s.t. $N_1$ and $N_2$ are disjoint, where $P_1 \cap P_2 = \varnothing$ and $T_1 \cap T_2 = \varnothing$. 
	Let $P_c$ be a set of channels. A channel-composition of $N_1$ and $N_2$, denoted $N_1 \oplus_{P_c} N_2$, is a Petri net $N = (P, T, F, m_0, m_f)$, where:
	\vspace{-0.1cm}
	\begin{enumerate}
		\item $P = P_1 \cup P_2 \cup P_c$, where $P_c \cap (P_1 \cup P_2) = \varnothing$.
		\item $T = T_1 \cup T_2$.
		\item $F = F_1 \cup F_2 \cup F_c$, where $F_c \subseteq (P_c \times (T_1 \cup T_2))\cup ((T_1 \cup T_2) \times P_c)$.
		\item $m_0 = m_0^1 \cup m_0^2$ and $m_f = m_f^1 \cup m_f^2$.
		\item $\forall p \in P_c \colon$
		\begin{enumerate}
			\item $(\pre{p} \subseteq T_1$ or $\pre{p} \subseteq T_2)$ and $(\post{p} \subseteq T_1$ or $\post{p} \subseteq T_2)$,
			\item $\pre{p} \subseteq T_i \Leftrightarrow \post{p} \subseteq T_{(i+1)\bmod 2}$,
			\item $\pre{p} \neq \varnothing$ and $\post{p} \neq \varnothing$.
		\end{enumerate}
	\end{enumerate}
\end{definition}

\begin{remark}\label{netplusprop}
	The operation $\oplus_{P_c}$
	is \emph{commutative}, i.e. $N_1 \oplus_{P_c} N_2 = N_2 \oplus_{P_c} N_1$. 
\end{remark}

By the following proposition, the class of GWF-nets is closed under the channel-composition. Figure \ref{alpha_ex}(a) provides an example of channel-composition of $N_1$ and $N_2$, where channels are indicated by small gray places. 

\begin{proposition} \label{closed}
	If $N_1$ and $N_2$ are GWF-nets, then $N_1 \oplus_{P_c} N_2$ is a GWF-net.
\end{proposition}
\vspace{-0.4cm}
\begin{proof}
	We show that each channel lies on a path from a place in $m_0^i$ to a place in $m_f^j$, where $j={(i+1)\bmod 2}$ and $i=1, 2$.  Take $p \in P_c$. By Def. \ref{netplus}.5, $\pre{p} \subseteq T_i \Leftrightarrow \post{p} \subseteq T_j$. Take $t_i \in \pre{p}$ and $t_j \in \post{p}$. By Def. \ref{WF}.3, $\exists s \in m_0^i$, s.t. $(s, t_i) \in F_i^*$. Then $(s, p) \in F^*$. By Def. \ref{WF}.3, $\exists f \in m_f^j$, s.t. $(t_j, f) \in F_j^*$. Then $(p, f) \in F^*$.   \qed
\end{proof}

Further, we recall the definition of $\alpha$-morphisms (see Definition 6 and Definition 7 in \cite{Bernardinello2013}) supporting abstraction and refinement for Petri nets. An example of $\alpha$-morphism is shown in Fig. \ref{alpha_ex}(b), where the refinement is given by the shaded ovals and by the transition labeling explicitly. 
Refinement can also require splitting transitions of the abstract model.
After giving the formal definition of $\alpha$-morphisms, we will provide an intuition behind them.
\begin{definition}\label{alpha}
	Let $N_i = (P_i, T_i, F_i, m_0^i)$ be a marked SMD safe Petri net, $X_i = P_i \cup T_i$, $i=1, 2$. 
	An $\alpha$-morphism from $N_1$ to $N_2$ is a total surjective map $\varphi \colon X_1 \to X_2$, also denoted $\varphi\colon N_1 \to N_2$, such that:
	\begin{enumerate}
		\item $\varphi(P_1) = P_2$.
		\item $\varphi(m_0^1) = m_0^2$.
		\item $\forall t \in T_1 \colon$ if $\om{t} \in T_2$, then $\om{\pre{t}}=\pre{\om{t}}$ and $\om{\post{t}}=\post{\om{t}}$.
		\item $\forall t \in T_1 \colon$ if $\om{t} \in P_2$, then $\om{\neighb{t}}=\{\om{t}\}$.
		\item $\forall p_2 \in P_2 \colon$
		\begin{enumerate}
			\item $N_1(\varphi^{-1}(p_2))$ is an acyclic net;
			\item $\forall p \in \inp{N_1(\varphi^{-1}(p_2))} \colon \om{\pre{p}} \subseteq \pre{p_2}$, and if $\pre{p_2} \neq \varnothing$, then $\pre{p} \neq \varnothing$;
			\item $\forall p \in \outp{N_1(\varphi^{-1}(p_2))} \colon \om{\post{p}} = \post{p_2}$;
			\item $\forall p \in P_1 \cap \varphi^{-1}(p_2) \colon p \notin \inp{N_1(\varphi^{-1}(p_2))} \Rightarrow \om{\pre{p}}=p_2 \text{ and }$ \\$p \notin \outp{N_1(\varphi^{-1}(p_2))} \Rightarrow \om{\post{p}} = p_2$;
			\item $\forall p \in P_1 \cap \varphi^{-1}(p_2)\colon$ there is a sequential component $N' = (P', T', F')$ of $N_1$, s.t. $p \in P'$ and $\varphi^{-1}(\neighb{p_2}) \subseteq T'$.
		\end{enumerate}
	\end{enumerate}
\end{definition}

\begin{figure}
	\centering
	\vspace{-0.7cm}
	\subfigure[$N_1 \oplus_{P_c} N_2$, $P_c=\{x, y, z\}$]{\includegraphics[height=4.8cm]{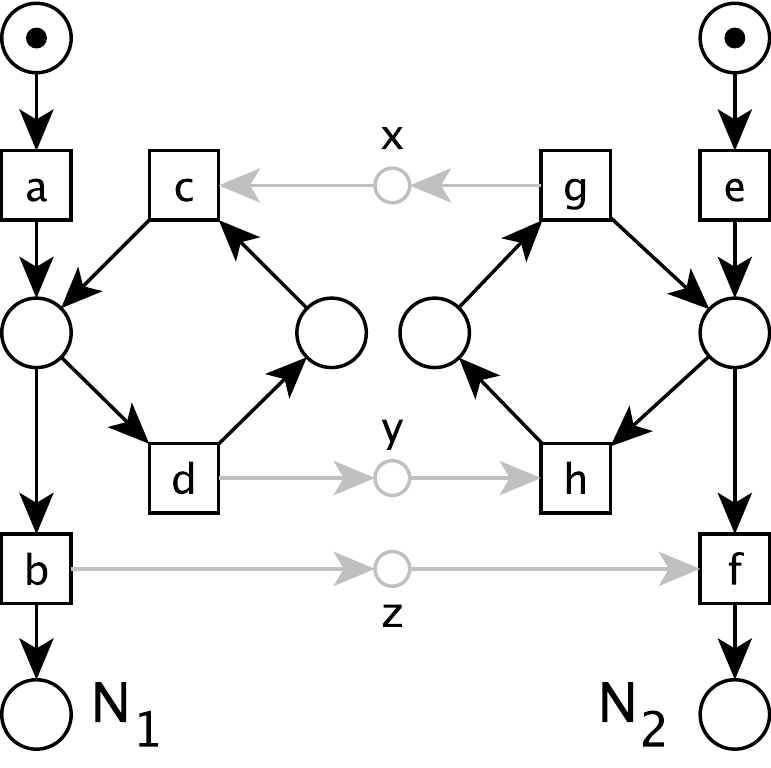} \label{channelcomp}}\hspace{0.5cm}
	\subfigure[$\alpha$-morphism $\varphi \colon N_2' \to N_2$]{\includegraphics[height=5.44cm]{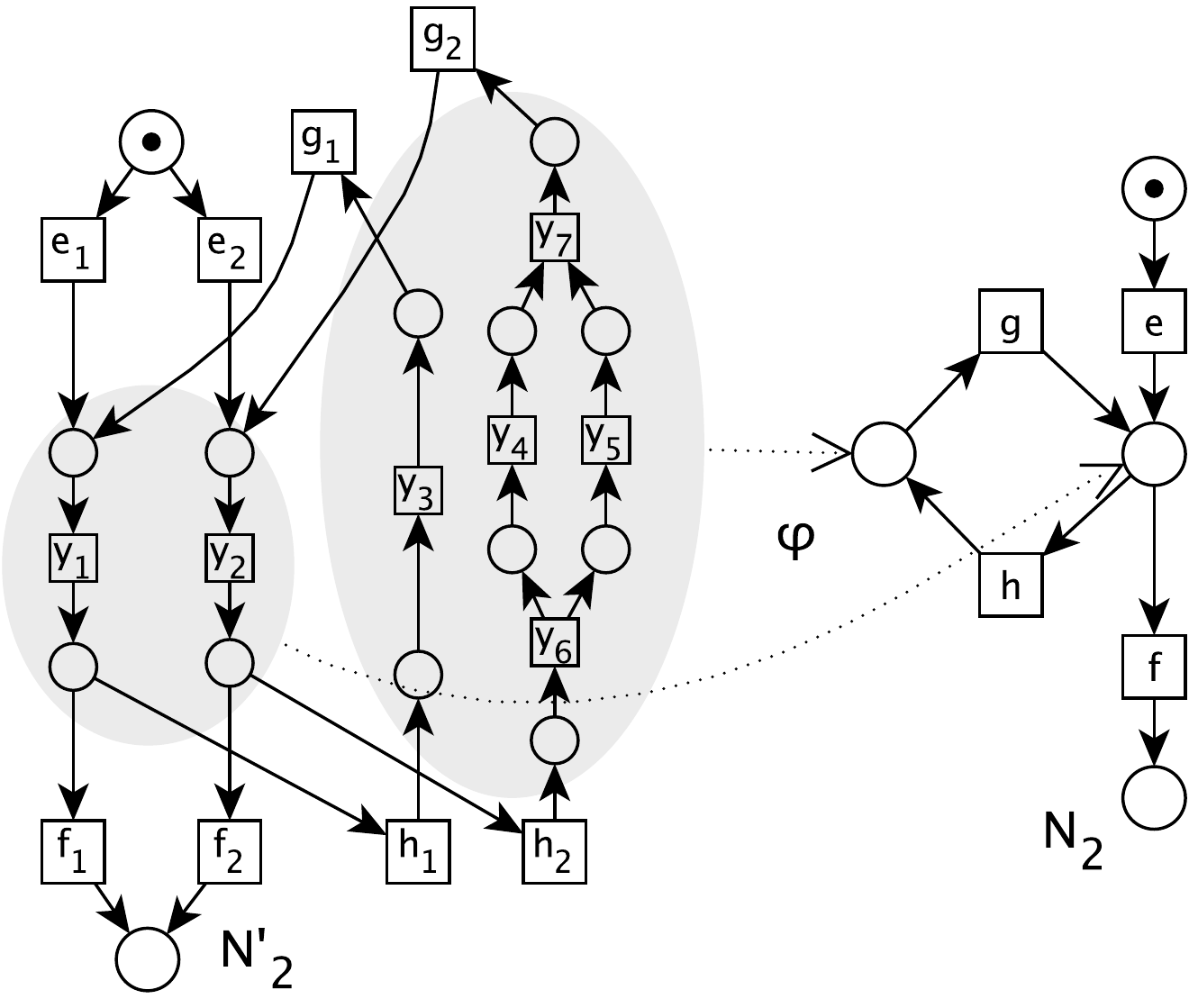}}
	\caption{Introductory examples}
	\label{alpha_ex}
	\vspace{-0.4cm}
\end{figure}

By definition, $\alpha$-morphisms allow us to refine places by replacing them with subnets. Thus, if a transition is mapped on a place, then its neighborhood should be mapped on the same place (4). If a transition is mapped on a transition, it should have the corresponding neighborhood (3). 

Indeed, $\alpha$-morphisms are motivated by the attempt to ensure that properties of an abstract model hold in its refinement. Each output place of a subnet should have the same choices as its abstraction does (5c). However, input places do not need this constraint (5b), because the choice between them is done before, since there are no concurrent events in the neighborhood of the subnet\,(5e). Moreover, 5d guarantees that presets and postsets of places, internal to the subnet, are mapped to the same place as the subnet. Conditions 5a-5e together ensure the intuition behind $\alpha$-morphisms. If a subnet of $N_1$ refines a place in $N_2$, then it behaves ``in the same way''. More precisely, by Lemma 1 of \cite{Bernardinello2013}, (a) no input transition of the subnet is enabled whenever a token is inside the subnet; and (b) firing an output transition of the subnet empties it.

\section{Properties Preserved and Reflected by $\alpha$-Morphisms}
In this section, we study properties preserved and reflected by $\alpha$-morphisms. In \cite{Bernardinello2013} several properties of $\alpha$-morphisms have already been studied. Here we will mention some of them and consider properties of $\alpha$-morphisms for GWF-nets.
\medskip

In the following propositions, we assume $N_i = (P_i, T_i, F_i, m_0^i)$ to be a marked SMD safe Petri net for $i=1, 2$, s.t. there is an $\alpha$-morphism $\varphi \colon X_1 \to X_2$, where $X_i = P_i \cup T_i$.
\medskip

To begin with, $\alpha$-morphisms preserve the structure of GWF-nets.

\begin{proposition}\label{WFpreser}
	If $N_1$ is a GWF-net, then $N_2$ is a GWF-net.
\end{proposition}

\begin{proof}
	We prove that $N_2$ satisfies three structural conditions of Def. \ref{WF}.
				
		\textbf{1.} By Def. \ref{alpha}.2, $\om{m_0^1} = m_0^2$. Suppose $\exists p_2 \in m_0^2 \colon \pre{p_2} \neq \varnothing$. By Def. \ref{alpha}.5b, $\forall p \in \inp{N_1(\varphi^{-1}(p_2))} \colon$ if $\pre{p_2} \neq \varnothing$, then $\pre{p} \neq \varnothing$. Take $p_1 \in m_0^1$, s.t. $\om{p_1} = p_2$. Since $p_1 \in \inp{N_1(\varphi^{-1}(p_2))}$, then $\pre{p_1} \neq \varnothing$. By Def. \ref{WF}.1, $\forall p \in m_0^1 \colon \pre{p} = \varnothing$. Then, $\pre{p_2} =\varnothing$ and $\forall p \in m_0^2 \colon \pre{p} = \varnothing$. 
			
		\textbf{2.}	By Def. \ref{WF}.2, $m_f^1 \subseteq P_1$, s.t. $\forall p \in {m_f^1} \colon \post{p} = \varnothing$. Denote $\om{m_f^1}$ by $m_f^2 \subseteq P_2$. Suppose $\exists p_2 \in m_f^2 \colon \post{p_2} \neq \varnothing$. Take $p_1 \in m_f^1$, s.t. $\varphi(p_1) = p_2$. By Def. \ref{alpha}.5c, $\forall p \in \outp{N(\varphi^{-1}(p_2))} \colon \om{\post{p}} = \post{p_2}$. Since $p_1 \in \outp{N(\varphi^{-1}(p_2))}$, then $\post{p_1} \neq \varnothing$. But $p_1 \in m_f^1$ and $\post{p_1} = \varnothing$ (by Def. \ref{WF}.2). Then $\post{p_2} = \varnothing$ and $\forall p \in {m_f^2} \colon \post{p}=\varnothing$.
			
		\textbf{3.}	Suppose $\exists x_2 \in X_2$, s.t. $\forall p \in m_0^2 \colon (p, x_2) \notin F_2^*$. By Def. \ref{alpha}, $\varphi^{-1}(x_2) = \{x_1^1, \dots, x_1^k\} \subseteq X_1$. If $x_2 \in T_2$, then $\varphi^{-1}(x_2) \subseteq T_1$, and take $x_1 \in \varphi^{-1}(x_2)$. If $x_2 \in P_2$, then take $x_1 \in \inp{N(\varphi^{-1}(x_2))}$. By Def. \ref{WF}.3, $\exists s \in m_0^1 \colon (s, x_1) \in F_1^*$. Then, $\om{\pre{x_1}} \in \pre{x_2}$ or $\om{\pre{x_1}} = x_2$. We come backward through the whole path from $s$ to $x_1$ mapping it on $N_2$ with $\varphi$. We obtain that $\exists x' \in X_2 \colon (x', x_2) \in F_2^*$, s.t. $\om{s} = x'$. 
			
		Suppose $\exists x_2 \in X_2$, s.t. $\forall p \in m_f^2 \colon (x_2, p) \notin F_2^*$. By Def. \ref{alpha}, $\varphi^{-1}(x_2) = \{x_1^1, \dots, x_1^k\} \subseteq X_1$. If $x_2 \in T_2$, then $\varphi^{-1}(x_2) \subseteq T_1$, and take $x_1 \in \varphi^{-1}(x_2)$. If $x_2 \in P_2$, then take $x_1 \in \outp{N(\varphi^{-1}(x_2))}$. By Def. \ref{WF}.3, $\exists f \in m_f^1 \colon (x_1, f) \in F_1^*$. Then, $\om{\post{x_1}} \in \post{x_2}$ or $\om{\post{x_1}} = x_2$. We come forward through the whole path from $x_1$ to $f$ mapping it on $N_2$ with $\varphi$. We obtain that $\exists x' \in X_2 \colon (x_2, x') \in F_2^*$, s.t. $\om{f} = x'$. \qed
\end{proof}

\begin{remark}It follows from Proposition \ref{WFpreser} that $\om{m_f^1} = m_f^2$.
	In the general case the converse of Proposition \ref{WFpreser} is not true. In fact,
	$\alpha$-morphisms do not reflect the initial state of GWF-nets properly (see Fig. \ref{gwf_nonrefl}).
\end{remark}

\begin{figure}
	\centering
	\vspace{-0.5cm}
	\subfigure[]{\includegraphics[height=4cm]{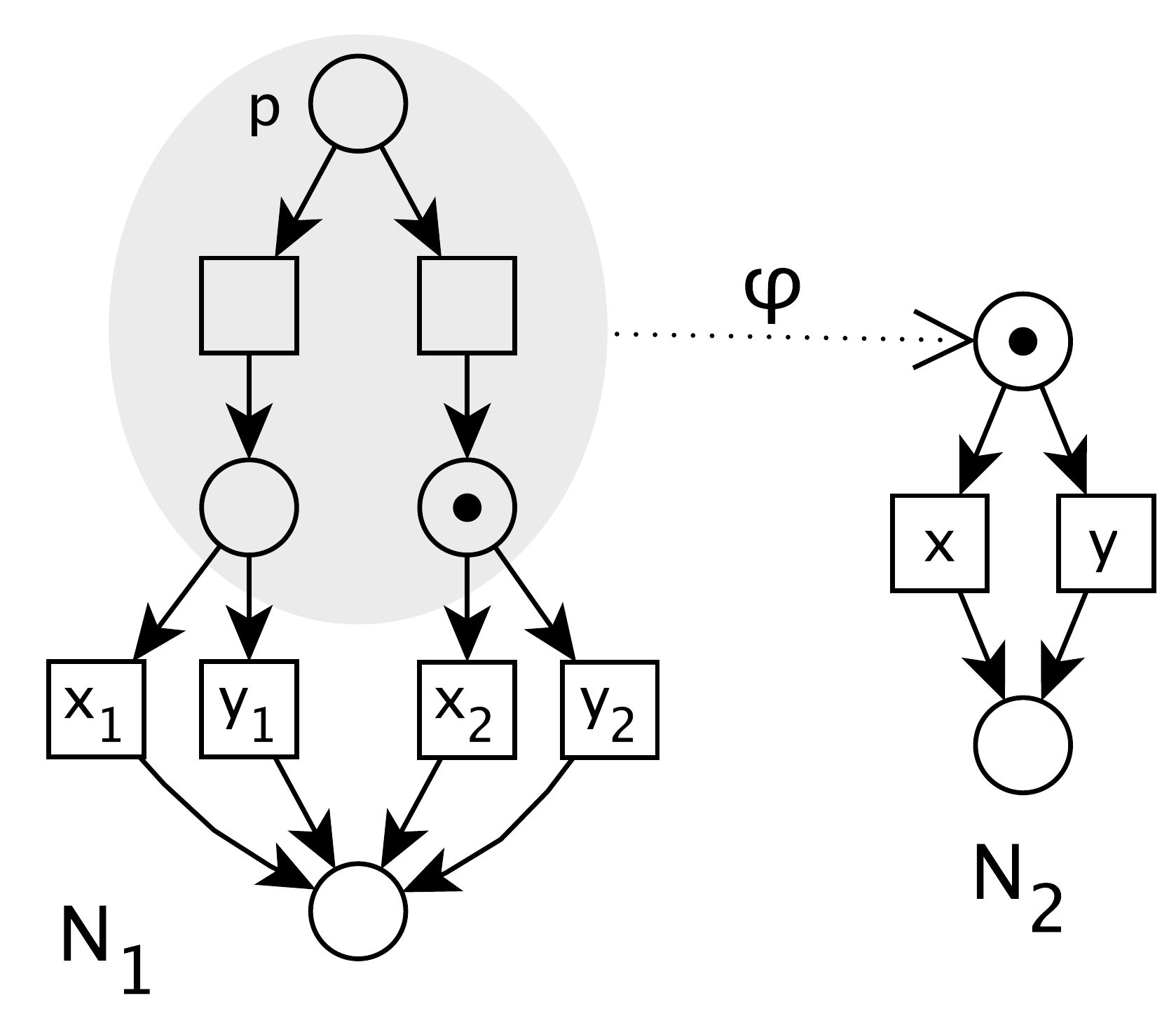} \label{gwf_nonrefl}}
	\subfigure[]{\includegraphics[height=4cm]{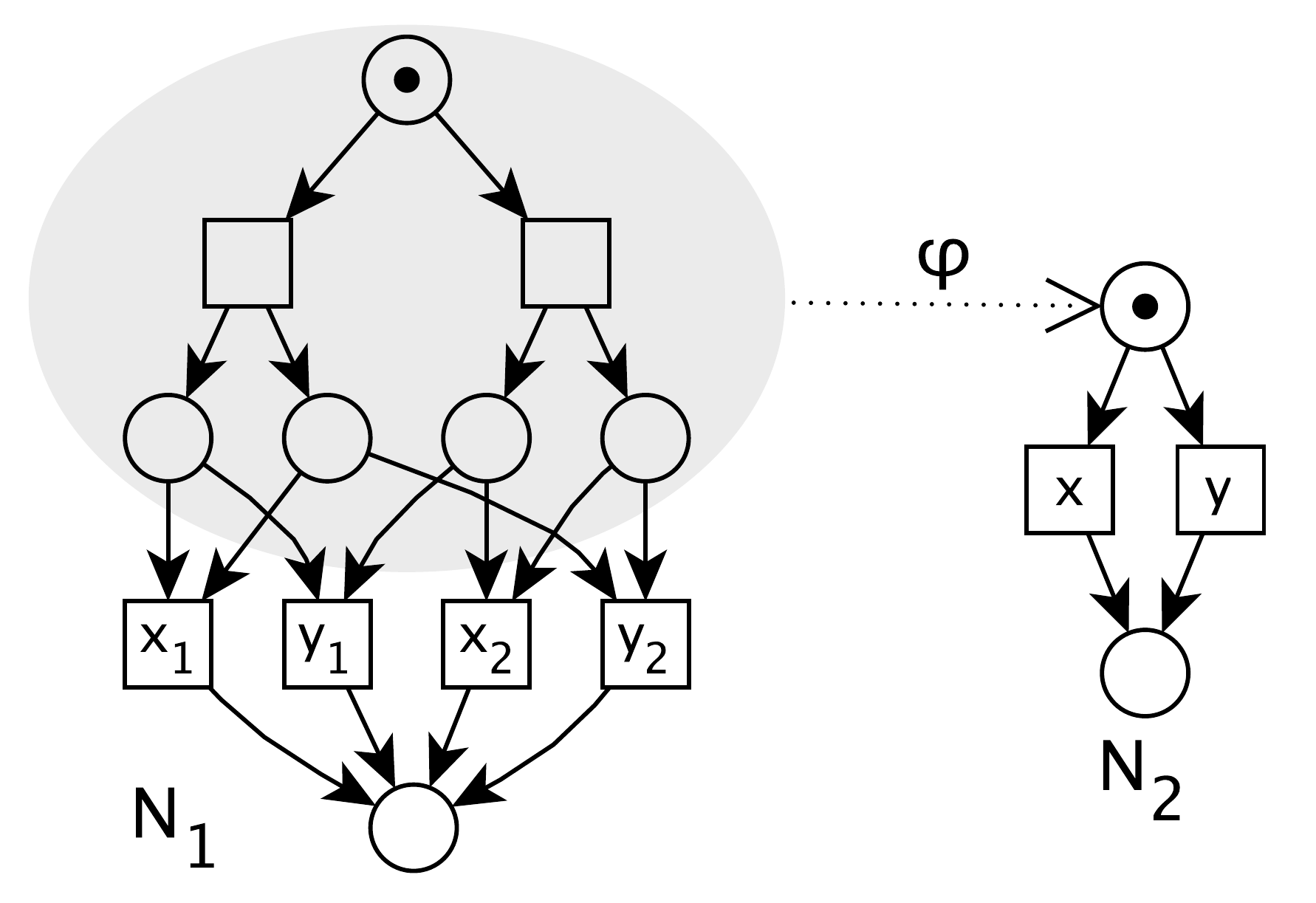} \label{sound_nonrefl}}
	\caption{An $\alpha$-morphism: two examples}
	\vspace{-0.4cm}
\end{figure}

Recall from \cite{Bernardinello2013} that $N_1$ is \emph{well marked} w.r.t. $\varphi$ if each input place of a subnet in $N_1$,
refining a marked place in $N_2$, is marked. 
Consider the $\alpha$-morphism shown in Fig. \ref{gwf_nonrefl}, the token of the shaded subnet must be placed into $p$ to make $N_1$ well marked w.r.t. to $\varphi$. 
The structure of GWF-nets is reflected under the well-markedness of $N_1$ (see the following Proposition). 
However, if $N_1$ is a GWF-net related to $N_2$ by an $\alpha$-morphism, then $N_1$ is well marked w.r.t. $\varphi$.
\newpage

\begin{proposition}\label{wfrefl}
	If $N_2$ is a GWF-net and $N_1$ is well marked w.r.t. $\varphi$, then $N_1$ is a GWF-net.
\end{proposition}
\begin{proof}
		We prove that $N_1$ satisfies three structural conditions of Def. \ref{WF}.
	
		\textbf{1.} By Def. \ref{WF}.1, $\forall s_2 \in m_0^2 \colon \pre{s_2}=\varnothing$. 
		Since $N_1$ is well marked w.r.t. $\varphi$, $m_0^1=\{\inp{N_1(\varphi^{-1}(s_2))}\,\vert\, s_2 \in m_0^2\}$. 
		Take $N_1(\varphi^{-1}(s_2))$ corresponding to $s_2 \in m_0^2$. Assume $\exists p \in \inp{N_1(\varphi^{-1}(s_2))}$, s.t. $\pre{p} \neq \varnothing$. Note that $\om{p} =s_2$. Then $\om{\pre{p}} = s_2$ (by Def. \ref{alpha}.4), and $p \notin \inp{N_1(\varphi^{-1}(s_2))}$. 
		
		\textbf{2.} By Def. \ref{WF}.2, $\forall f_2 \in m_f^2 \colon \post{f_2}=\varnothing$. Take $N_1(\varphi^{-1}(f_2))$ corresponding to $f_2 \in m_f^2$ and $p \in \outp{N_1(\varphi^{-1}(f_2))}$. Note that $\om{p} = f_2$. Assume $\post{p} \neq \varnothing$. Then $\om{\post{p}} = p_2$ (by Def. \ref{alpha}.4) and $p \notin \outp{N_1(\varphi^{-1}(f_2))}$. We obtain that $m_f^1 = \{\outp{N_1(\varphi^{-1}(f_2))}\,\vert\, f_2 \in m_f^2\}$ and $\forall f_1 \in m_f^1 \colon \post{f_1} = \varnothing$.
		
		\textbf{3. } Suppose $\exists x_1 \in X_1$, s.t. $\forall s_1 \in m_0^1 \colon (s_1, x_1) \notin F_1^*$. If $(x_1, x_1) \notin F_1^*$, we go backward through the path from $x_1$ to the first node $x_1' \in X_1$, s.t. $\pre{x_1'} = \varnothing$. 
		Since $\forall t \in T_1 \colon \abs{\pre{t}} \geq 1$, then $x_1' \in P_1$. If $x_1' \notin m_0^1$, then $N_1$ is not well marked w.r.t. $\varphi$. 
		If $(x_1, x_1) \in F_1^*$, we know by Def. \ref{alpha}.5a that there is a corresponding cycle in $N_2$.
		Take $x_2 \in X_2$, s.t. $\om{x_1} = x_2$. 
		By Def. \ref{WF}.3, $\exists s_2 \in m_0^2 \colon (s_2, x_2) \in F_2^*$. 
		Take $x_2' \in X_2$ belonging to this cycle, s.t. at least one element in $\pre{x_2'}$ is not in the cycle. Since $\varphi$ is surjective, $\exists x_1' \in X_1 \colon \om{x_1'} = x_2'$, belonging to the cycle $(x_1, x_1) \in F_1^*$. 
		If $x_2' \in T_2$, then $\varphi^{-1}(x_2') \subseteq T_1$. 
		By Def. \ref{alpha}.3, the neighborhood of transitions is preserved by $\varphi$.
		Then, $\forall t_1 \in \varphi^{-1}(x_2') \colon \om{\pre{t_1}} = \pre{x_2'}$, i.e. there must be a place in $\pre{\varphi^{-1}(x_2')}$ which is not in the cycle $(x_1, x_1) \in F_1^*$. 
		If $x_2' \in P_2$, then take $\inp{N_1(\varphi^{-1}(x_2'))}$. At least one place in $\inp{N_1(\varphi^{-1}(x_2'))}$ should have an input transition which is not in the cycle $(x_1, x_1) \in F_1^*$, since there exists an element in $\pre{x_2'}$ which is not in the cycle in $N_2$. We have shown that $\exists x \in \pre{x_1'}$, s.t. $x$ is not in the cycle $(x_1, x_1) \in F_1^*$. Now either there is a path from $x$ to $x^*$ in $N_1$, s.t. $\pre{x^*} = \varnothing$, or there is another cycle $(x^*, x^*) \in F_1^*$.
		
		By using a similar reasoning, we prove that $\forall x_1 \in X_1 \,\exists f_1 \in m_f^1 \colon (x_1, f_1) \in F_1^*$. The only difference is that we need to go forward through the paths.\qed 
\end{proof}

We recall in the following proposition that $\alpha$-morphisms preserve reachable markings and
firing of transitions. 
\newpage

\begin{proposition}[\cite{Bernardinello2013}]\label{markpres}
	Let $m_1 \in \reach{m_0^1}$. Then $\om{m_1} \in \reach{m_0^2}$. If $m_1 \reach{t}m_1'$, where $t \in T_1$, then:
	\begin{enumerate}
		\item $\om{t} \in T_2 \Rightarrow \om{m_1}\reach{\om{t}}\om{m_1'}$.
		\item $\om{t} \in P_2 \Rightarrow \om{m_1} = \om{m_1'}$.
	\end{enumerate}
\end{proposition}

\begin{remark}[\cite{Bernardinello2013}]
	In the general case $\alpha$-morphisms do not reflect reachable markings, i.e. $m_2 \in \reach{m_0^2}$  and $m_2 \reach{t_2}$ do not imply that there exists $m_1 = \varphi^{-1}(m_2) \in \reach{m_0^1}$, s.t. $ \forall t \in \varphi^{-1}(t_2) \colon m_1\reach{t}$.
\end{remark}

To reflect reachable markings, we need to check local conditions as shown in \cite{Bernardinello2013}. For any place $p_2 \in P_2$, refined by a subnet of $N_1$, we construct a ``local'' net, denoted $S_2(p_2)$, of $N_2$ by taking the neighborhood transitions of $p_2$ plus an artificial input and output place if needed. At the same time, we construct the corresponding ``local'' net, denoted $S_1(p_2)$, of $N_1$ by taking the subnet of $N_1$ refining $p_2$, i.e. $N_1 (\varphi^{-1}(p_2))$, and the transitions $\varphi^{-1} (\pre{p_2})$, $\varphi^{-1} (\post{p_2})$ plus an artificial input and output place if needed.
The details are given in \cite{Bernardinello2013} (see Definition 9 there).

There is an $\alpha$-morphism $\varphi^S$ from $S_1(p_2)$ to $S_2(p_2)$ which is a restriction of $\varphi$ on the places and transitions of $S_1(p_2)$. 
In the following Lemma, taking the unfolding of $S_1(p_2)$, we assure that the final marking of the subnet enables the same set of transitions which are enabled by its image.
By Definition \ref{alpha}.5a, since $S_1(p_2)$ is acyclic, its unfolding is finite.


\begin{lemma}\label{unf}
	Let $\mathcal{U}(S_1(p_2))$ be the unfolding of $S_1(p_2)$ with the folding function $u$, and $\varphi^S$ be an $\alpha$-morphism from $S_1(p_2)$ to $S_2(p_2)$.  Let $N_1$ be a sound GWF-net. Then, the map from $\mathcal{U}(S_1(p_2))$ to $S_2(p_2)$ obtained as $\varphi^S \circ u$ is an $\alpha$-morphism.

\end{lemma}
	\begin{proof}
	Since $N_1$ is a GWF-net, $S_1(p_2)$ is a GWF-net. By \cite{Bernardinello2013} (see Lemma 1 there), when a transition in $\varphi^{-1}(\post{p_2})$ fires, it empties the subnet $N_1(\varphi^{-1}(p_2))$. 
Then $S_1(p_2)$ is sound, and by Def. \ref{sound}.3 $\forall t \in T_1 \exists m \in \reach{m_0^1} \colon m\reach{t}$. So, each transition of $S_1(p_2)$ will occur at least once. Then, the folding $u$ is a surjective function from $\mathcal{U}(S_1(p_2))$ to $S_1(p_2)$, and $\varphi^S \circ u$ is an $\alpha$-morphism from $\mathcal{U}(S_1(p_2))$ to $S_2(p_2)$.
	\qed
	\end{proof}


As for the main results, we obtain that under the assumption of soundness of the GWF-net $N_1$, $\alpha$-morphisms not only preserve, but also reflect reachable markings and transition firings (see Proposition \ref{markrefl}). Moreover, $\alpha$-morphisms preserve soundness as shown in Proposition \ref{spres}.



\begin{proposition}\label{markrefl}
	If $N_1$ is a sound GWF-net, then $\forall m_2 \in \reach{m_0^2}\, \exists m_1 \in \reach{m_0^1} \colon$ $\varphi(m_1) = m_2$, and if $\exists t_2 \in T_2 \colon m_2\reach{t_2}$, then $\forall t_1 \in \varphi^{-1}(t_2) \colon m_1 \reach{t_1}$.
\end{proposition}

\begin{proof}
	Follows from Lemma \ref{unf} and from \cite{Bernardinello2013} (see Proposition 5 there).
\end{proof}


\begin{proposition}\label{spres}
	If $N_1$ is a sound GWF-net, then $N_2$ is a sound GWF-net.
\end{proposition}

	\begin{proof}
		We prove that $N_2$ satisfies three behavioral conditions of Def. \ref{sound}.
				
		\textbf{1.} By Def. \ref{sound}.1, $\forall m_1 \in \reach{m_0^1} \colon m_f^1 \in \reach{m_1}$. Then, $\exists w \in T_1^* \colon m_1 \reach{w} m_f^1$, i.e. $w=t_1t_2\dots t_n$ and $m_1\reach{t_1}m_1^1\dots m_1^{n-1}\reach{t_n}m_f^1$. By Prop. \ref{markpres}, we can simulate $w$ on $N_2$, and by Prop. \ref{WFpreser}, $\om{m_f^1} = m_f^2$. Now assume $\exists m_2 \in \reach{m_0^2} \colon m_f^2 \notin \reach{m_2}$. By Prop. \ref{markrefl}, $\exists m_1' \in \reach{m_0^1} \colon \varphi^{-1}(m_2)=m_1'$. By Def. \ref{sound}.1, $m_f^1 \in \reach{m_1'}$, i.e. $\exists w' \in T^* \colon m_1' \reach{w'} m_f^1$. By Prop. \ref{markpres}, we again simulate $w'$ on $N_2$. Then $m_f^2 \in \reach{m_2}$.
		
		\textbf{2.}	Suppose $\exists m_2' \in \reach{m_0^2} \colon m_2' \supseteq m_f^2$. Then we can write $m_2' = m_f^2 \cup P'$, where $\forall p \in P' \colon p \notin m_f^2$. By Prop. \ref{markrefl}, take $m_1 \in \reach{m_0^1}$, s.t. $\varphi^{-1}(m_1) = m_2'$ and $m_f^1 \not\subseteq m_1$. By Def. \ref{sound}.1, $m_1^f \in \reach{m_1}$, i.e. $\exists w \in T_1^* \colon m_1 \reach{w}m_1^f$. By Prop. \ref{markpres}, we simulate $w$ on $N_2$. By Prop. \ref{WFpreser}, $\om{m_f^1} = m_f^2$. The only way to completely empty places in $P'$ is to consume at least one token from $m_f^2$. Then $\exists p \in m_f^2 \colon \post{p} \neq \varnothing$ which is a contradiction.
				
		\textbf{3.} By Def. \ref{sound}.3, $\forall t_1 \in T_1 \,\exists m_1 \in \reach{m_0^1} \colon m_1 \reach{t_1}$. The map $\varphi$ is surjective, i.e. $\forall t_2 \in T_2 \colon \exists t_1 \in T_1 \colon \om{t_1} = t_2$. By Prop. \ref{markpres}, $m_1\reach{t_1}m_1' \Rightarrow \om{m_1}\reach{\om{t_1}}\om{m_1'}$. Then, $\forall t_2 \in T_2 \, \exists m_2 \in \reach{m_0^2} \colon m_2 \reach{t_2}$. \qed
	\end{proof}

\begin{remark}
	In the general case the converse of Proposition \ref{spres} is not true. Consider the example shown in Fig. \ref{sound_nonrefl}, where $N_2$ is sound and $N_1$ is not sound, since transitions $y_1$ and $y_2$ are dead. Thus, $\alpha$-morphisms do not reflect soundness. Note that reachable markings are also not reflected in this example. However, in the next section we will provide conditions under which soundness is reflected.
\end{remark}

\section{From Event Logs to Structured and Sound Models of Multi-Agent Systems}

In this section, we present our approach to process discovery by composing individual agent models through $\alpha$-morphisms. An \emph{event log} $L$ is a finite multiset of finite non-empty sequences (\emph{traces}) over a set of \emph{observable} actions $\act$.

%

We assume to have an event log $L$ of two interacting agents. For instance, we will further work with the same event log obtained from the MAS shown in Fig. \ref{exintro}(a) which we have used in Section 1. We assume to know what actions are executed by which agent, $\act = A_1 \cup A_2$, s.t. $A_1 \cap A_2 = \varnothing$. Also, we assume to know actions corresponding to their asynchronous ``message-passing'' interaction.

Instead of discovering the model directly from $L$, we propose to filter the log according to the agent actions $A_1$ and $A_2$ producing two new logs $L_1$ and $L_2$. Traces of $L_1$ and $L_2$ contain actions done only by a corresponding agent. By using, for example, \emph{inductive miner} \cite{Inductive13}, from $L_1$ and $L_2$ we obtain two GWF-nets $N_1'$ and $N_2'$ modeling the two agents. By construction, $N_1'$ and $N_2'$ are \emph{well-structured} and sound, which implies that they are state machine decomposable (see Corollary~4 in \cite{Aalst00}). Well-structured models are recursively built from blocks representing basic control flow constructs, e.g. choice, concurrency or cycle.

It is possible to compose $N_1'$ and $N_2'$ using the channel-composition as in Definition \ref{netplus} obtaining a new GWF-net $N_1' \oplus_{P_c} N_2'$, where channels are defined according to the specification of agent interaction we have assumed to know.
However, it is obvious that $N_1' \oplus_{P_c} N_2'$ might not be sound. 
In order to avoid the verification of $N_1' \oplus_{P_c} N_2'$, 
we apply the following approach 
to get its soundness \emph{by construction}.

We can abstract the discovered nets $N_1'$ and $N_2'$ w.r.t. interacting actions, thus producing two GWF-nets $N_1$ and $N_2$, s.t. there is an $\alpha$-morphism $\varphi_i \colon N_i' \to N_i$ for $i=1, 2$. 
According to Proposition \ref{spres}, $N_i$ is sound.
The abstract models can also be composed via the same channels obtaining $N_1 \oplus_{P_c} N_2$. 
This abstract model represents the interaction protocol between the agents. 
Obviously, $N_1 \oplus_{P_c} N_2$ is less complex than $N_1' \oplus_{P_c } N_2'$, and its soundness can be much easier verified.

Given $N_1 \oplus_{P_c} N_2$ and the two $\alpha$-morphisms $\varphi_i \colon N_i' \to N_i$ for $i=1, 2$, we can construct two new GWF-nets: $N_1' \oplus_{P_c} N_2$ and $N_1 \oplus_{P_c} N_2'$, which actually represent different abstractions of the same MAS. This construction is formally defined in the following definition. We also show in Remark \ref{channelalpha} that there is an $\alpha$-morphism from $N_1' \oplus_{P_c} N_2$ (by symmetry, from $N_1 \oplus_{P_c} N_2'$) towards $N_1 \oplus_{P_c} N_2$. 

\begin{definition}\label{refchancomp}
	Let $N_i = (P_i, T_i, F_i, m_0^i, m_f^i)$ be a GWF-net for $i=1, 2$, and $P_c$ be a set of channels. 
	Let $N_1 \oplus_{P_c} N_2= (P, T, F, m_0, m_f)$ be a channel-composition of $N_1$ and $N_2$. Let $N_1' = (P_1', T_1', F_1', m'_{01}, m'_{f1})$ be a GWF-net, s.t. there is an $\alpha$-morphism $\varphi_1 \colon N_1' \to N_1$. Construct $N_1' \oplus_{P_c} N_2 = (P', T', F', m_0', m_f')$, where:
	\vspace{-0.1cm}
	\begin{enumerate}
		\item $P' = P_1' \cup P_2 \cup P_c$.
		\item $T' = T_1' \cup T_2$.
		\item $F' = F_1' \cup F_{N_2}^c \cup F_{N_1'}^c$, where:
		\begin{enumerate}
			\item $F_{N_2}^c = F \cap [((P_c \cup P_2)\times T_2) \cup (T_2 \times (P_c \cup P_2))]$;
			\item $F_{N_1'}^c \subseteq (P_c \times T_1') \cup (T_1' \times P_c)$.
		\end{enumerate}
		\item $m_0' = m'_{01} \cup m_0^2$ and $m_f' = m'_{f1} \cup m_f^2$.
		\item $\forall p \in P_c$, $\forall t\in T_1 \colon$ \\ 
		$((t, p) \in F \Rightarrow \varphi^{-1}(t) \times \{p\} \subseteq F_{N_1'}^c)$ and 
		$((p, t) \in F \Rightarrow \{p\} \times \varphi^{-1}(t)\subseteq F_{N_1'}^c)$.
	\end{enumerate}
\end{definition}

\begin{remark}\label{channelalpha}
	Let $N_1$, $N_1'$ and $N_2$ be GWF-nets, and $\varphi_1 \colon N_1' \to N_1$ be an $\alpha$-morphism. 
	Then there is an $\alpha$-morphism $\varphi'_1 \colon N_1' \oplus_{P_c} N_2 \to N_1 \oplus_{P_c} N_2$. 
	In fact, by construction, $\varphi'_1$ is given by $\varphi_1$ plus the identity mapping of places and transitions of $N_2$ together with the identity mapping of channel places.
\end{remark}

\begin{example}
	Here we consider as $N_1 \oplus_{P_c} N_2$ the GWF-net shown in Fig.\ref{alpha_ex}(a). We can refine it by models $N_1'$ or $N_2'$, discovered from filtered event logs $L_1$ and $L_2$ as shown in Fig. \ref{N'oplusN}. The $\alpha$-morphisms between the discovered models and the abstract ones are indicated by the shaded ovals and the transition labeling. The $\alpha$-morphism between $N_2'$ and $N_2$ has also been shown in Fig. \ref{alpha_ex}(b).
\end{example}

As for the main result, we will prove in Proposition \ref{main} that an $\alpha$-morphism from $N_1' \oplus_{P_c} N_2$ (by symmetry, from $N_1 \oplus_{P_c} N_2'$) to $N_1 \oplus_{P_c} N_2$ reflects the soundness of $N_1 \oplus_{P_c} N_2$. To prove this fact, we will use the property of reachable markings of a channel-composition stated in the following Lemma.

\begin{lemma}\label{markdecomp}
	Let $N_i = (P_i, T_i, F_i, m_0^i, m_f^i)$ a GWF-net for $i=1, 2$, $N_1 \oplus_{P_c} N_2$ $= (P, T, F, m_0, m_f)$.
	Then, $\forall m \in \reach{m_0} \colon m = m_1 \cup m_c \cup m_2$, where $m_1 \in \reach{m_0^1}, m_2 \in \reach{m_0^2}, \text{and } m_c \subseteq P_c$. 
\end{lemma}

\begin{proof}
	By Def. \ref{netplus}.4, $m_0 = m_0^1 \cup m_0^2$ and $m_f = m_f^1 \cup m_f^2$. Take $m \in \reach{m_0}$, then $\exists w \in FS(N_1 \oplus_{P_c} N_2) \colon m_0\reach{w}m$, where $w \in T^*$. 
	By Def. \ref{netplus}.2, $T=T_1 \cup T_2$. Restricting $w$ to $T_1$ and $T_2$ produces two firing sequences of $N_1$ and $N_2$ leading from $m_0^1$ and $m_0^2$ to the reachable markings $m_1$ and $m_2$ which constitute $m$. \qed
\end{proof}

\begin{figure}[H]
	\centering
	\vspace{-0.7cm}
	\subfigure[$N_1' \oplus_{P_c} N_2$]{\includegraphics[height=5.5cm]{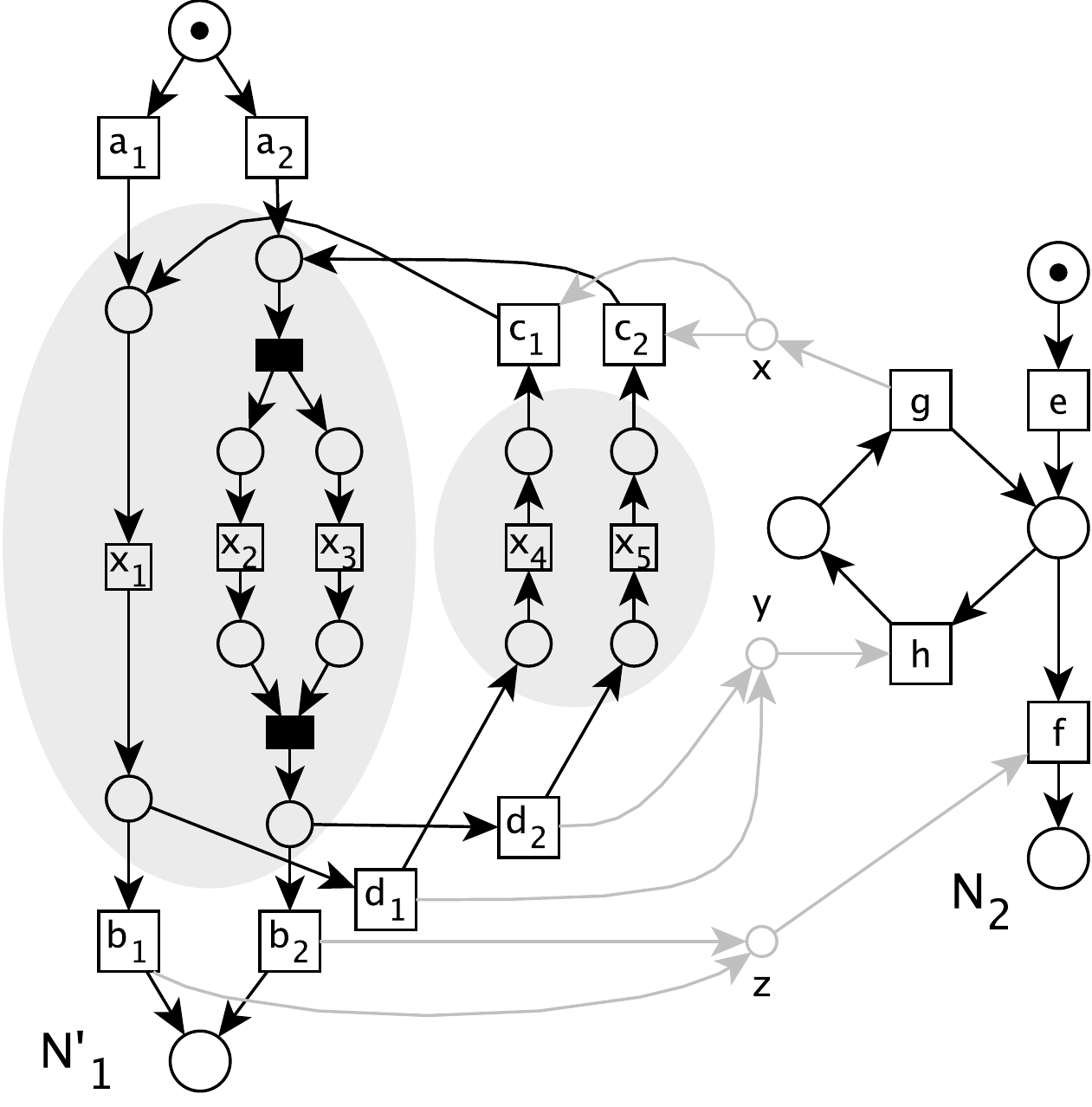} \label{channelcompref}}\hspace{0.5cm}
	\subfigure[$N_1 \oplus_{P_c} N_2'$]{\includegraphics[height=5.52cm]{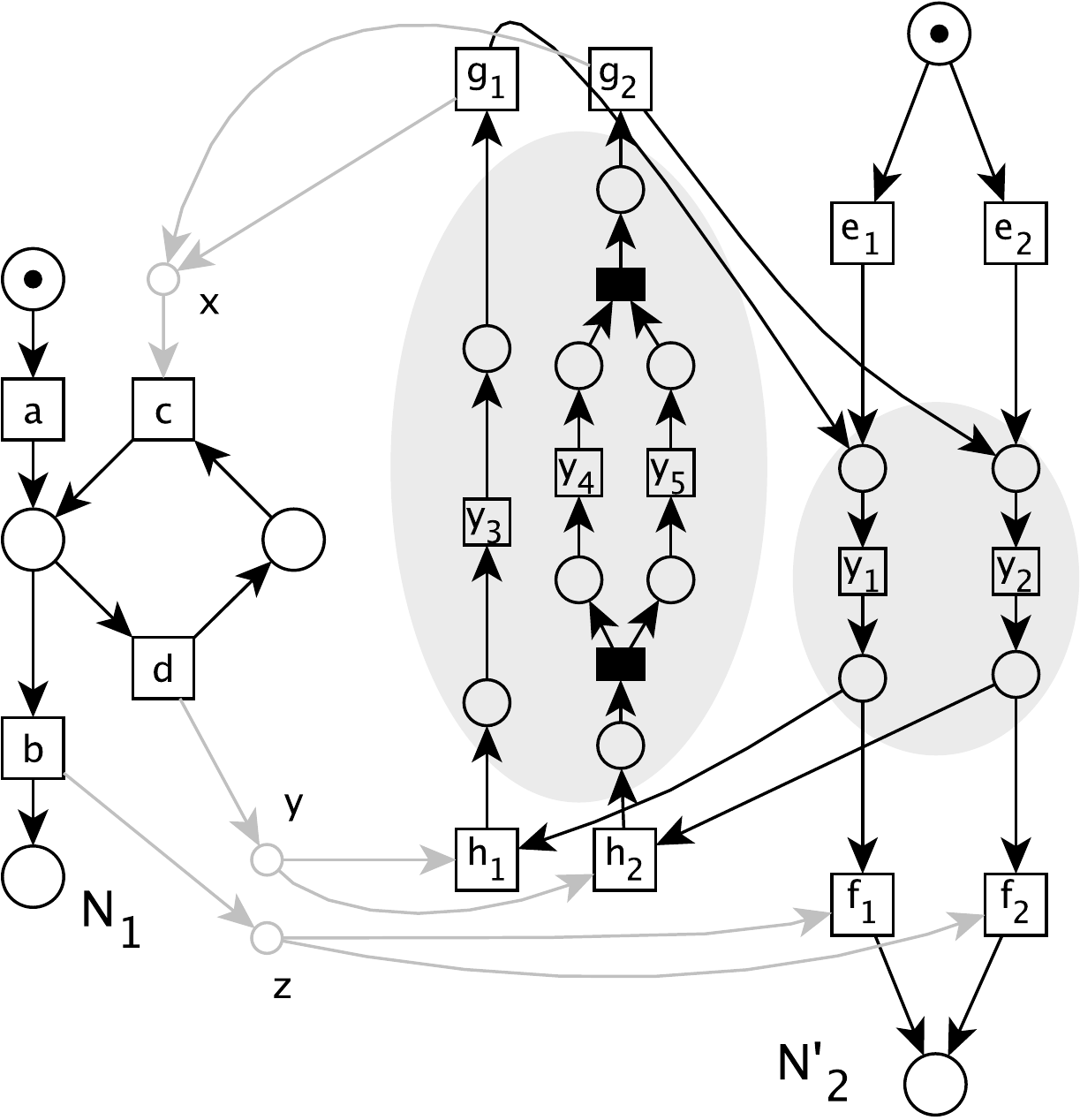} \label{channelcompref2}}
	\caption{Refining $N_1 \oplus_{P_c} N_2$ with agent nets $N_1'$ and $N_2'$ discovered from filtered logs}
	\label{N'oplusN}
	\vspace{-0.4cm}
\end{figure}

\begin{proposition}\label{main}
	Let $N_1, N_1'$ and $N_2$ be sound GWF-nets, and $\varphi_1 \colon N_1' \to N_1$ be an $\alpha$-morphism. If $N_1 \oplus_{P_c} N_2$ is sound, then $N_1' \oplus_{P_c} N_2$ is sound.
\end{proposition}

\begin{proof}
	By Rem. \ref{channelalpha}, there is an $\alpha$-morphism $\varphi_1'$ from $N'_1\oplus_{P_c} N_2$ to $N_1 \oplus_{P_c} N_2$.
	Assume $N_1' \oplus_{P_c} N_2 = (P', T', F', m_0', m_f')$ and $N_1 \oplus_{P_c} N_2 = (P, T, F, m_0, m_f)$.
	We prove that $N'_1 \oplus_{P_c} N_2$ satisfies the three conditions of soundness of Def. \ref{sound}. 
	
		\textbf{1.} Take $m' \in \reach{m_0'}$. 
		By Lemma \ref{markdecomp}, $m' = m_1' \cup m_2 \cup m_c$.
		By Prop. \ref{markpres} for $\varphi'_1$, $\varphi'_1(m') = m \in \reach{m_0}$. 
		By Lemma \ref{markdecomp}, $m=m_1 \cup m_2 \cup m_c$, where $m_2$ and $m_c$ are the same as in $m'$ and $\varphi_1(m_1') = m_1$ (by Prop. \ref{markpres} for $\varphi_1$).
		Since $N_1 \oplus_{P_c} N_2$ is sound, $\exists w \in FS(N_1 \oplus_{P_c} N_2) \colon m\reach{w}m_f$, where $w \in T^*$.
		It is possible to write $w=w_2^1v$, where $v=\epsilon$ or $v=t_1^1w^2_2t_1^2\dots$ with $w_2^i \in T_2^*$ and $t_1^i \in T_1$, s.t. $i \geq1$.
		Then $w_2^i$ can be obviously executed on the component $N_2$ of $N_1'\oplus_{P_c} N_2$ as well, because $\varphi_1'$ reflects the connection to channels (by Def. \ref{refchancomp}.5 and Rem. \ref{channelalpha}).
		Since 
		$N_1'$ is sound, $\varphi_1$ reflects reachable markings and firing of transitions (by Prop.~\ref{markrefl}) between $N_1'$ and $N_1$.
		Thus, there is a reachable marking $m'_{1i}$ in $N_1'$, belonging to $\varphi_1^{-1}(m_1^i)$ in $N_1$, s.t. if $m_1^i \reach{t_1^i}$ in $N_1$, then any transition in $\varphi_1^{-1}(t_1^i)$ is enabled at $m'_{1i}$ in $N_1'$.
		Moreover, these transitions are enabled in $N_1' \oplus_{P_c} N_2$, because $\varphi_1'$ reflects connection to channels (by Def. \ref{refchancomp}.5 and Rem. \ref{channelalpha}).
		Hence, the sequence $w \in FS(N_1 \oplus_{P_c} N_2)$ is reflected in $N_1' \oplus_{P_c} N_2$ reaching its final marking $m_f'$.
		
		\textbf{2.} Suppose by contradiction $\exists m' \in \reach{m_0'} \colon m' \supseteq m_f'$ and $m' \neq m_f'$. 
		By Def. \ref{refchancomp}, $m_f' = m_{f1}' \cup m_f^2$. 
		Then $m' = m_{f1}' \cup m_f^2 \cup m_3$.
		By Prop. \ref{markpres} for $\varphi_1'$, $\varphi'_1(m') \in \reach{m_0}$.
		In this way, $\varphi_1'(m') = \varphi_1'(m_{f1}') \cup \varphi_1'(m_f^2) \cup \varphi_1'(m_3) = \varphi_1(m_{f1}') \cup m_f^2 \cup m_3 = m_f^1 \cup m_f^2 \cup m_3 = m_f \cup m_3$ (by Rem. \ref{channelalpha} and Def. \ref{refchancomp}.4).
		Thus, this marking strictly covers the final marking $m_f$ of $N_1 \oplus_{P_c} N_2$ which contradicts its soundness.
		
		\textbf{3.} We prove that $\forall t' \in T' \,\exists m' \in \reach{m_0'} \colon m'\reach{t'}$.
		By Lemma \ref{markdecomp}, $m'=m_1'\cup m_2 \cup m_c$.
		By Def. \ref{refchancomp}.2, $\forall t' \in T' \colon t' \in T_1'$ or $t' \in T_2$. 
		If $t' \in T_2$, then since $N_1 \oplus_{P_c} N_2$ is sound, $\exists m \in \reach{m_0}$, s.t. $m \reach{t'}$. 
		By Rem. \ref{channelalpha} and by Def. \ref{refchancomp}.5, $m_2\cup m_c$ enables $t'$. 
		If $t' \in T_1'$, then there are two cases. 
		If $\varphi_1'(t') \in P$, then $t'$ is not connected to channels.
		Since $N_1'$ is sound, $m'_1 \subseteq m'$ enables $t'$.
		If $\varphi_1'(t') \in T$, then take $t \in T$, s.t. $\varphi_1'(t') = t$. 
		Since $N_1 \oplus_{P_c} N_2$ is sound, $\exists m \in \reach{m_0} \colon m \reach{t}$. 
		By Prop.~\ref{markrefl}, $\varphi_1$ reflects reachable markings and firings of transitions on $N_1'$. 
		Moreover, $\varphi_1'$ reflects connection to channels (by Def. \ref{refchancomp}.5 and Rem. \ref{channelalpha}). 
		Then $m_1' \cup m_c$ enables\,$t'$.   		 \qed



\end{proof}


Since there are two $\alpha$-morphisms from $N_1 \oplus_{P_c} N_2'$ and $N_1' \oplus_{P_c} N_2$ towards $N_1 \oplus_{P_c} N_2$, we can compose them by using the composition defined in \cite{Bernardinello2013} (see Definition 12) obtaining as a result $N$ with other two $\alpha$-morphisms from $N$ towards $N_1 \oplus_{P_c} N_2'$ and $N_1' \oplus_{P_c} N_2$, s.t. the diagram shown in Fig. \ref{final}(b) commutes.

Alternatively, by applying a similar construction as the one given in Definition \ref{refchancomp}, we can refine $N_2$ in $N_1' \oplus_{P_c} N_2$ by $N_2'$ obtaining $N_1' \oplus_{P_c} N_2'$ which is isomorphic to the previously obtained composition $N$. Symmetrically,  it is possible to refine $N_1$ in $N_1 \oplus_{P_c} N_2'$ by $N_1'$ coming to the same result. According to Remark \ref{channelalpha}, there are $\alpha$-morphisms from $N_1' \oplus_{P_c} N_2'$ to $N_1' \oplus_{P_c} N_2$ and to $N_1 \oplus_{P_c} N_2'$. 
According to Proposition \ref{main}, $N_1' \oplus_{P_c} N_2$ (by symmetry, $N_1 \oplus_{P_c} N_2'$) is sound, and then $N_1' \oplus_{P_c} N_2'$ is also sound.

Thus, we have shown that it is possible to even simultaneously refine $N_1$ and $N_2$ in the sound abstract model $N_1 \oplus_{P_c} N_2$ by sound models $N_1'$ and $N_2'$ obtained from filtered event logs, s.t. the result $N_1' \oplus_{P_c} N_2'$ is also sound.

\begin{example}
	In Fig. \ref{final}(a), we show the result of composing, by means of $\alpha$-morphisms, $N_1' \oplus_{P_c} N_2$ and $N_1 \oplus_{P_c} N_2'$ constructed in Example 1. 
	This composition can also be obtained directly by the channel-composition of $N_1'$ and $N_2'$. 
	The obtained model meets the desired requirement stated in Section 1: we can identify the agents as subnets and explicitly see their asynchronous interaction. 
\end{example}



\begin{figure}[H]
	\centering
	\vspace{-0.6cm}
	\subfigure[composition $N$ isomorphic to $N'_1 \oplus_{P_c} N_2'$]{\includegraphics[height=5.6cm]{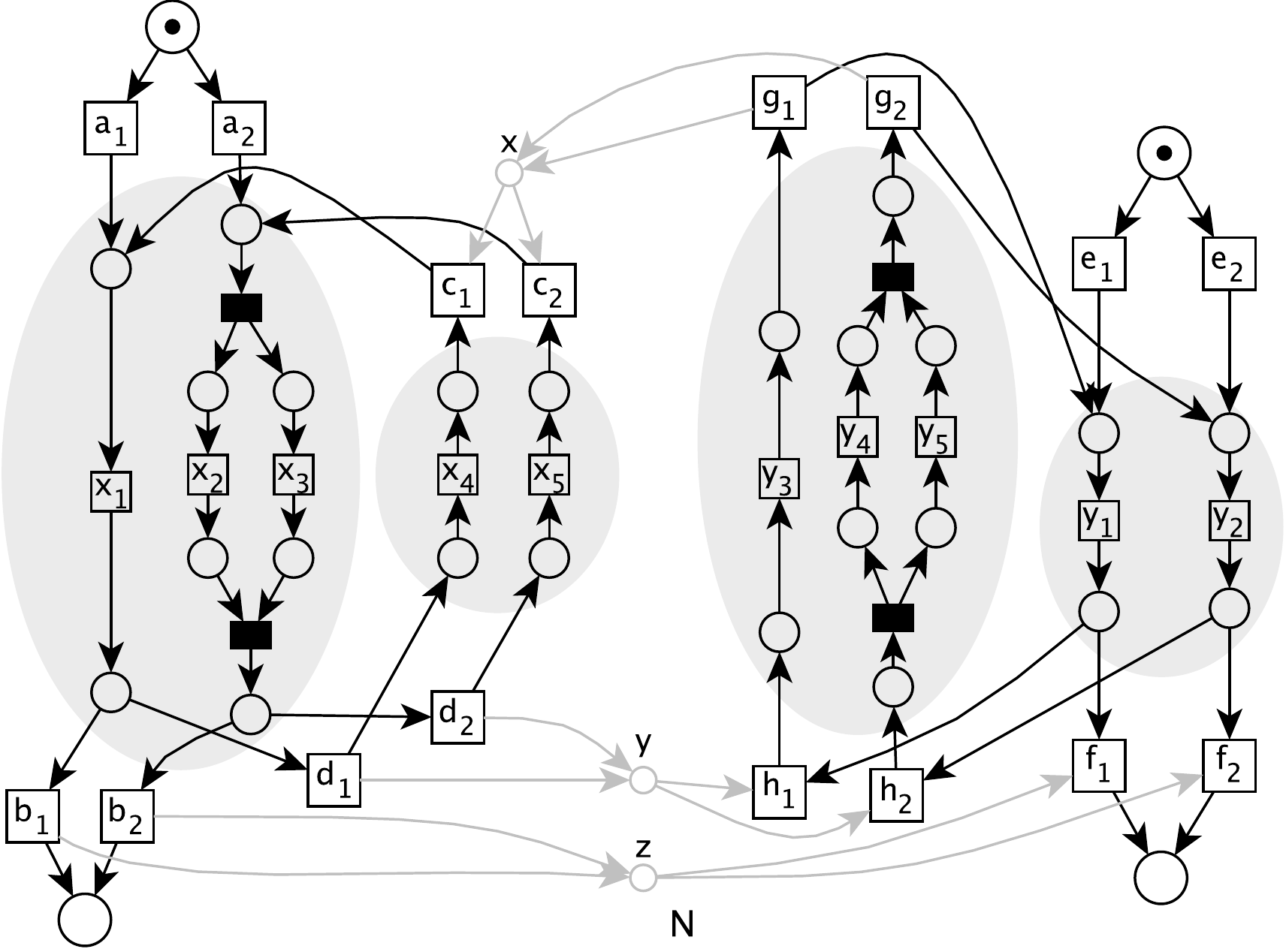}}
	\subfigure[diagram]{\includegraphics[width=4.5cm]{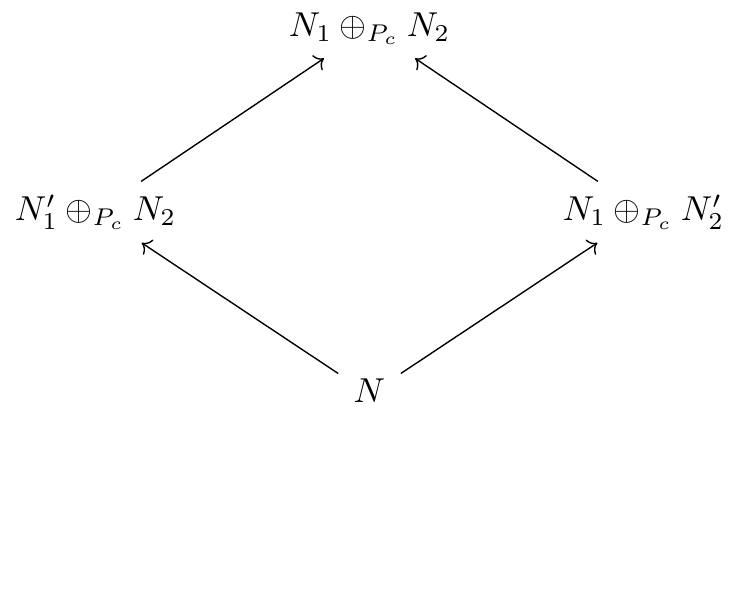}}
	\caption{Composition of $N_1' \oplus_{P_c} N_2$ and $N_1 \oplus_{P_c} N_2'$ based on $\alpha$-morphisms} \label{final}
	\vspace{-0.9cm}
\end{figure}

\subsection{Comparing Quality of Presented MAS Process Models}
	Process models of MAS presented in Fig. \ref{exintro}(b) and in Fig. \ref{exintro}(c) are discovered directly from the log produced by simulating the model shown in Fig. \ref{exintro}(a). The composition shown in Fig. \ref{final}(a), in fact, is obtained from the same log. That is why, we have compared their quality using the standard process discovery metrics \cite{Quality12}. \emph{Fitness} measures how accurately a model can replay traces of an initial event log. Intuitively, \emph{precision} indicates a ratio between the behavior given by the log and the one allowed by the model. If it is low, then a model allows for too much additional behavior. 
	Table \ref{maintab} provides the result of this quality analysis.
	
	\begin{table}[h]
		\vspace{-0.3cm}
		\centering
		\caption{Fitness and precision of MAS models presented in the paper}
		\label{maintab}
		\begin{tabular}{@{}lllcccc@{}}
			\toprule
				\textbf{Model} & \textbf{Algorithm} & \textbf{Discovery} &\phantom{s}& \textbf{Fitness} &\phantom{s}& \textbf{Precision} \\
			\midrule
				Figure \ref{exintro}(b) & Inductive miner & Direct &&1,0000&&0,1732 \\
				Figure \ref{exintro}(c) & ILP miner & Direct &&1,0000&&0,8516 \\
				Figure \ref{final}(a) & Inductive miner & Composed &&1,0000&&0,8690 \\
			\bottomrule
		\end{tabular}
		\vspace{-0.3cm}
	\end{table}
	The result of the compositional process discovery shows the increase in precision resulting from the separate analysis of agent behavior. However, we can see that the precision of the composed model is close to that of the model obtained by ILP miner.
	Thus, the model can be appropriate in terms of precision, but not in terms of the MAS structure.
	
\section{Conclusion and Future Work}

In this paper, we have proposed a compositional approach to discover process models of multi-agent systems from event logs.
We have considered asynchronous agent interactions which are modeled by a channel-composition of two nets.
We assume that events can be partitioned, s.t. system event logs are filtered according to agent behavior to discover their detailed models.
In order to guarantee that their composition is sound, we have proposed to abstract agent models, by using $\alpha$-morphisms, w.r.t. interacting actions thus constructing an abstract model of a communication protocol. 
The general algorithm of constructing $\alpha$-morphisms is one of the open problems which is a subject for further research.
We have proven that, when this abstract protocol is sound, the direct composition of the detailed models is also sound.
In proving this fact, we have used two intermediate models corresponding to a detailed view of one agent composed with the abstract view of the other.

The obtained system model is structured in such a way that it is possible to identify agent models as components and their interactions are clearly expressed.
We have compared the quality of directly obtained models with the quality of models obtained by the proposed approach. 
The quality of a composed model is seen to be at the appropriate level in terms of fitness and precision being the general process discovery quality dimensions.
In the future, we suggest introducing a structural indicator showing the extent to which it is easy to identify agents as parts of the model as well as their interaction.

We plan to explore more general asynchronous interactions, e.g. when the same channel can be used by several agents as well as the possibility to have both asynchronous and synchronous communication. Moreover, we would like to generalize the proposed approach to other classes of nets. Another line of research will be focused on dealing with dynamic changes. In particular, intermediate models, mentioned above, can be used to check if changes in the behavior of an agent affect the soundness of the whole system.


%
%
%
%
%
%
		
	\bibliographystyle{splncs03}
	\bibliography{reflist}

\newpage

\end{document}